\documentclass[11pt]{article}
\usepackage{fullpage}
\usepackage[linesnumbered,ruled]{algorithm2e}
\usepackage{empheq} 
\usepackage{cases} 

\usepackage{amssymb}
\usepackage{amsthm}
\usepackage{amsmath}
\usepackage{graphicx}
\newtheorem{theorem}{Theorem}[section]
\newtheorem{lemma}{Lemma}[section]

\newtheorem{claim}{Claim}[section]
\newtheorem{proposition}{Proposition}[section]
\newtheorem{definition}{Definition}[section]

\newcommand{\cA}{{\cal A}}
\newcommand{\cB}{{\cal B}}
\newcommand{\cD}{{\cal D}}
\newcommand{\cR}{{\cal R}}
\newcommand{\cS}{{\cal S}}
\newcommand{\cT}{{\cal T}}
\DeclareMathOperator{\Penalty}{\texttt{Penalty}}
\newcommand{\overhead}{\mathcal{O}}
\newcommand{\KnowType}{\mathcal{T}}
\DeclareMathOperator{\dist}{yes}
\DeclareMathOperator{\nodist}{no}

\DeclareMathOperator{\blind}{blind}
\DeclareMathOperator{\comp}{comp}

\newcommand{\KTdb}{\KnowType_{\blind}^{\dist}}
\newcommand{\KTnb}{\KnowType_{\blind}^{\nodist}}
\newcommand{\KTdc}{\KnowType_{\comp}^{\dist}}
\newcommand{\KTnc}{\KnowType_{\comp}^{\nodist}}

\begin{document}
\bibliographystyle{plain}
\title{{\bf Impact of Knowledge on the Cost\\ of Treasure Hunt in Trees }}
\author{
S\'{e}bastien Bouchard\footnotemark[1]
\and
Arnaud Labourel\footnotemark[2]
\and
Andrzej Pelc \footnotemark[3] \footnotemark[4]
}
\date{ }
\maketitle
\def\thefootnote{\fnsymbol{footnote}}

\footnotetext[1]{ 
\noindent Univ. Bordeaux, Bordeaux INP, CNRS, LaBRI, UMR5800, Talence, France.\\
E-mail: {\tt sebastien.bouchard@u-bordeaux.fr}
}
\footnotetext[2]{Aix Marseille Univ, Universit\'e de Toulon, CNRS, LIS, Marseille, France.\\
E-mail: {\tt arnaud.labourel@lis-lab.fr}}
\footnotetext[3]{ \noindent
D\'epartement d'informatique, Universit\'e du Qu\'ebec en Outaouais, Gatineau,
Qu\'ebec J8X 3X7, Canada.\\
E-mail: {\tt pelc@uqo.ca}
}
\footnotetext[4]{ \noindent
Partially supported by NSERC discovery grant 2018-03899 and
by the Research Chair in Distributed Computing at the
Universit\'e du Qu\'{e}bec en Outaouais.
}

\begin{abstract}
		A mobile agent has to find an inert target in some environment that can be a graph or a terrain in the plane. This task is known as {\em treasure hunt}. We consider deterministic algorithms for treasure hunt in trees. Our goal is to establish the impact of different kinds of initial knowledge given to the agent on the cost of treasure hunt, defined as the total number of edge traversals until the agent reaches the treasure hidden in some node of the tree. The agent can be initially given either a {\em complete map} of the tree rooted at its starting node, with all port numbers marked, or a {\em blind map} of the tree rooted at its starting node but without port numbers. It may also be given, or not, the distance from the root to the treasure. This yields four different {\em knowledge types} that are partially ordered by their precision. (For example knowing the blind map and the distance is less precise than knowing the complete map and the distance). The {\em penalty} of a less precise knowledge type ${\cal T}_2$ over a more precise knowledge type ${\cal T}_1$  measures intuitively the worst-case ratio of the cost of an algorithm supplied with knowledge of type ${\cal T}_2$ over  the cost of an algorithm supplied with knowledge of type ${\cal T}_1$. Our main results establish penalties for comparable knowledge types in this partial order. For knowledge types with known distance, the penalty for having a blind map over a complete map turns out to be very large. By contrast, for unknown distance, the penalty of having a blind map over having a complete map is small. When a map is provided (either complete or blind), the penalty of not knowing the distance over knowing it is medium.

\noindent {\bf keywords:} mobile agent, treasure hunt, tree

\end{abstract}

	\section{Introduction}
	
	\subsection{The background} 
	
	Treasure hunt is the task of finding an inert target by a mobile agent in some environment that can be a graph or a terrain in the plane. This task has important applications when the terrain is dangerous or difficult to access for humans. Rescuing operations in mines contaminated or submerged by water are an example of situations where a lost miner is a target that has to be found fast by a mobile robot, and hence the length of the robot's trajectory should be as short as possible. In this case, the graph models the corridors of the mine with crossings representing nodes. Another example is searching for a data item in a communication network modeled by a graph.

			We consider deterministic algorithms for treasure hunt in trees. Apart from the above examples, trees can be an abstract model of a search space in task planing that requires finding search strategies operating within bounded memory \cite{DCD}, where the cost of backtracking is non-negligible. The goal is to find the solution represented as a node in the tree of possibilities and to do it as fast as possible.

\subsection{\bf The problem}
			
			 We will use the term {\em tree} for a tree rooted at the initial node of the agent, with port numbers at each node of degree $\delta$ arbitrarily numbered $0,\dots, \delta -1$. The treasure is hidden in some node of the tree and the agent has to find it. The agent is initially given some knowledge. It can be given either a {\em complete map}, with the root and all port numbers marked, or a {\em blind map} of the tree, with the root marked but without port numbers. It may also be given, or not, the distance from the root to the treasure. 
\footnote{Predefined port numbers are essential only if a complete map is given to the agent. In other cases the agent can assign its own port numbers as the tree is processed.}			 
			 
			 This yields four different {\em knowledge types} that are partially ordered by their precision. (For example knowing the blind map and the distance is less precise than knowing the complete map and the distance). These knowledge types will be denoted $\KTdc$, $\KTdb$, $\KTnc$, $\KTnb$, where the superscript indicates whether the agent is given the distance from the root to the treasure, and the subscript indicates the type of map given to the agent.

			Formally, any knowledge type $\cT$ is a family of sets $K$, such that each $K$, called {\em knowledge}, is a set of couples $(T,d)$ called {\em instances}, where $T$ is a tree and $d$ is a positive integer not larger than its depth. Each set $K$ is a possible input that an agent executing a treasure hunt algorithm can be given as initial knowledge. Thus, by getting knowledge $K$, the agent is given the set of possible instances that it may encounter. Suppose that the agent operates in some tree $T^*$ and the treasure is at distance $d^*$ from the root of this tree.  If $K\in \KTdc$, i.e., the algorithm is given a complete map and the distance, then $K$ is the one-element set $\{(T^*,d^*)\}$. If $K\in \KTdb$, i.e., the algorithm is given a blind map and the distance, then $K$ is the set of couples $(T,d^*)$, where $T$ are all possible trees resulting from $T^*$ by reassigning arbitrary port numbers at each node (this corresponds to giving the agent the blind map of the tree).  If $K\in \KTnc$,  i.e., the algorithm is given only a complete map, then $K$ is the set of all couples $(T^*,d)$, where $d$ are positive integers at most the depth of $T^*$, because these are the possible distances from the root to the treasure. Finally, if $K\in \KTnb$,  i.e., the algorithm is given only a blind map, then $K$ is the set of all couples $(T,d)$, where $T$ are all possible trees resulting from $T^*$ by reassigning arbitrary port numbers at each node, and $d$ are positive integers at most the depth of $T^*$.

			At each step of its execution, a deterministic algorithm for the agent has some history. At the beginning, the history is just a set $K$ given to the agent as initial knowledge. At each step, the agent sees all the port numbers at the current node. On the basis of its current history, the agent chooses one of these port numbers, traverses the corresponding edge and learns the port number by which it enters the adjacent node and learns the degree of this node (it sees all the port numbers at this node). It adds this information to its history and on the basis of it makes the next step. This proceeds until the agent finds the treasure.

			For any treasure hunt algorithm $\cA$ that is given initial knowledge $K$, and for any instance $(T, d)\in K$, we define the worst-case cost $C_{\cA}^K(T, d)$ of $\cA$ on $(T,d)$ as the total number of moves performed by $\cA$ (including possible revisits of nodes) until all nodes at distance exactly $d$ from the root have been visited. This value $C_{\cA}^K(T, d)$ is the worst-case cost of algorithm $\cA$ for finding the treasure at distance $d$ in the tree $T$, since the treasure can always be hidden in the last visited node at this distance.

			The quality of a treasure hunt algorithm $\cA$ is measured by comparing its cost to the distance between the root and the treasure. This distance is equal to the cost of the optimal algorithm having a complete map with the exact position of the treasure. This measure is done in the following way. For any integer $m\geq 1$ and for any knowledge $K$, let $B_m(K)=\{(T, d) \in K \mid d \leq m\}$ be the set of instances of $K$ for which the treasure is at a distance at most $m$. For any knowledge $K$ given to algorithm $\cA$, the ratio between $C_{\cA}^K(T, d)$ and $d$, maximized over all possible instances $(T, d)$ in $B_m(K)$ is called the \emph{overhead} $\overhead_{\cA}^K(m)$ of algorithm $\cA$ for knowledge $K$, with radius $m$, i.e.,

			\[\overhead_{\cA}^K(m) =
			\left\{
			\begin{array}{ll}
			\max\limits_{(T, d)\in B_m(K)}\dfrac{C_{\cA}^K(T, d)}{d}& \mbox{ if $B_m(K)\neq \emptyset$,}\\
			\\
			0 &\mbox{ otherwise.}
			\end{array}
			\right.
			\]

			Note that the above definition of overhead is more fine-grained than the classic competitive ratio approach, see, e.g.,  \cite{La}. In the spirit of the latter we would define the overhead of algorithm $\cA$  given knowledge $K$, maximizing the ratio over {\em all} possible instances in $K$, i.e.,  as $\cR^K_{\cA}=\max\limits_{(T, d)\in K}\frac{C_{\cA}^K(T, d)}{d}$. Instead, we define it for every radius $m$ separately. This approach is more precise, as some algorithms that seem to be quite good according to the traditional approach, turn out to be very inefficient for many values of $m$. As an example, consider the doubling strategy $\cD$ used in \cite{FKK+2008} on full binary trees. This strategy consists of iterative executions of DFS from the root to depths $2^i$, for $i=1,2,\dots$. It follows from \cite{FKK+2008} that, if $K$ is the knowledge consisting of all instances $(T,d)$ for which $T$ is a full binary tree of depth $h$ and $d$ is any positive integer at most $h$, then $\cR^K_{\cD}$ is at most 4 times larger than $\cR^K_{\cB}$ for the algorithm $\cB$ with the best possible $\cR^K_{\cB}$. However, this doubling strategy is really very bad, e.g., for radii $m$ of the form $2^k+1$. Indeed, for any such $m$, the nodes at distance $m$ from the root will only be visited using the doubling strategy when the agent performs the DFS at depth $2m-2=2^{k+1}$. In order to visit all nodes at distance $m$ with a DFS at depth $2m-2$, the agent needs to visit more than half of the nodes at distance at most $2m-2$ from the root and so $\overhead^K_{\cD}(m)\geq 2^{2m-2}/m$ (see Figure~\ref{fig_doublingstrategy}). For the algorithm $\cA$ consisting of iterative executions of DFS at $i=1,2,\dots$, we have $\overhead^K_{\cA}(m)\leq\sum_{i=1}^{m}2(2^{i+1}-2)/m$ since the subtree of depth $i$ contains $2^{i+1}-2$ edges and a DFS at depth $i$ traverses each of these edges twice. It follows that $\overhead^K_{\cA}(m)\leq\sum_{i=1}^{m}(2^{i+2}-4)/m\leq 2^{m+3}/m$ and so $\overhead^K_{\cD}(m) \geq 2^{m-5} \cdot \overhead^K_{\cA}(m)$. Thus for arbitrarily large radii $m$, the doubling strategy on trees is exponentially worse (according to our fine-grained definition of overhead) than the best possible algorithm. Our comparison between our notion of overhead and competitive ratio is also valid for the notion of search ratio defined in \cite{FKK+2008}. The main difference between the search ratio and competitive ratio is that the former is defined for a given fixed environment whereas in the latter the graph is not fixed and the agent has to learn it. Appart from this, these two definitions are similar. The doubling strategy is 4-search-competitive on trees\cite{FKK+2008} and so deemed good whereas with our fine-grained approach it is bad. In our considerations, in order to compare algorithms $\cA$ and $\cA'$ we compare $\overhead_{\cA}^K(m)$ to $\overhead_{\cA'}^K(m)$ for all $m$, and intuitively consider an algorithm $\cA$ to be better than $\cA'$ if it is better for all $m$.

			\begin{figure}
				\centering
				\includegraphics[width=14cm]{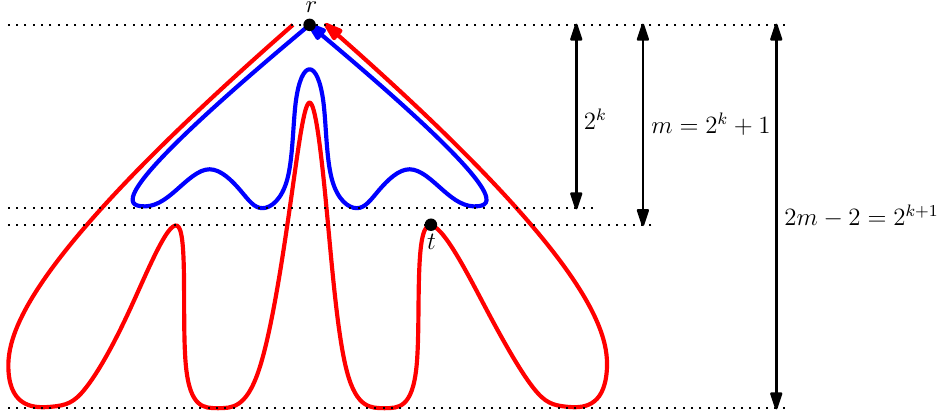}
				\caption{Illustration of a case for which the doubling strategy is bad since the treasure $t$ is only found after at least half a DFS at depth $2m-2$ is performed.\label{fig_doublingstrategy}}
			\end{figure}

			Let $T$ be a tree. We denote by $h(T)$ the depth of $T$, i.e., the length of the longest simple path from the root to a leaf of $T$. The $i$th level of the tree $T$ is the set of nodes at distance $i$ from the root. We define $l_d(T)$ as the number of nodes at the $d$th level, and $L_{d_1}^{d_2}(T) =\sum_{i=d_1}^{d_2}l_i(T)$ as the number of nodes at all levels from $d_1$ to $d_2$. For the sake of brevity, when there is no ambiguity on the tree, we will only write $l_d$ and $L_{d_1}^{d_2}$ for those values. Observe that these values can be computed from any map (complete or blind) of the tree $T$.

			There is a natural partial order $\ll$ of knowledge types, corresponding to  the decreasing precision of the initial knowledge given to the agent. For any pair of knowledge types $\cT_1$ and $\cT_2$, we have $\cT_1\ll\cT_2$ if for any $K_1\in \cT_1$ there exists $K_2\in \cT_2$ such that $K_1\subseteq K_2$. This means that knowledge $K_2$ of type $\cT_2$  is less precise than knowledge $K_1$ of type $\cT_1$.

			Our goal is to measure the impact of the quality of the initial knowledge given to treasure hunt algorithms on their overhead. To this end, we define the notion of penalty between two knowledge types $\KnowType_1\ll\KnowType_2$.

			\begin{definition}\label{def}
				Let $\KnowType_1\ll\KnowType_2$ be two knowledge types. For any function $f:\mathbb{N}\rightarrow\mathbb{N}$, we say that the penalty of $\KnowType_2$ over $\KnowType_1$, denoted $\Penalty(\KnowType_2/\KnowType_1)$, is $O(f)$ if there exists a constant $c>0$ and a deterministic algorithm $\cA_2$ such that for any deterministic algorithm $\cA_1$, we have:
				$$\forall K_2\in \KnowType_2,\forall m\geq 1,  \overhead_{\cA_2}^{K_2}(m) \leq c \cdot f(m) \cdot \max_{\substack{K_1\in \KnowType_1\\ K_1\subseteq K_2}}\left(\overhead_{\cA_1}^{K_1}(m)\right).$$
			\end{definition}

			The above definition captures our intention to measure the impact of knowledge on the performance of treasure hunt algorithms in the worst case up to radius $m$, for all positive integers $m$, and disregarding multiplicative constants. In fact, the worst-case aspect appears three times. First, it is already present in the definition of $C_{\cA}^K(T, d)$, for a given instance $(T,d)$, as this cost is defined as the cost of the algorithm for the worst-case location of the treasure at level $d$ of the tree $T$. Second, it appears in the definition of the overhead of an algorithm $\cA$ for given knowledge $K$ and any $m$, which is the worst-case ratio $C_{\cA}^K(T, d)/d$, over all instances $(T,d)$ in $B_m(K)$. Finally, the worst-case approach appears in the definition of penalty itself, as we compare the overhead of $\cA_2$, for any knowledge $K_2\in \cT_2$ and any $m$, to the overhead of any algorithm using a worst-case way of shrinking the set of instances in knowledge $K_2$ to get some knowledge of type $\cT_1$.

			For example, proving that the penalty of knowledge type $\KTnb$ over knowledge type $\KTnc$ is $O(1)$ means that there exists some constant $c>0$ and an algorithm that gets only a blind map of $T$, whose overhead is only $c$ times larger than the maximum overhead of an algorithm given a complete map resulting from this blind map by some port numbering, for any positive integer $m$.

			Hence the fact that the penalty of some knowledge type over a more precise knowledge type is small (resp. large) means that the corresponding difference of knowledge quality has a small (resp. large) impact on the overhead of treasure hunt algorithms.

			Notice that for three knowledge types $\cT_1 \ll \cT_2$ and $\cT_2 \ll \cT_3$, it may well happen that the penalty of $\cT_3$ over $\cT_1$ is smaller than the penalty of $\cT_2$ over $\cT_1$. This is because for the penalty of $\cT_2$ over $\cT_1$, the maximum overhead on the right-hand side of the inequality in Definition \ref{def} is taken over a smaller family of possible sets $K_1$, and hence it can be smaller than the maximum overhead taken for the penalty of $\cT_3$ over $\cT_1$. Actually, it turns out that such a situation occurs for knowledge types $\KTdc$, $\KTdb$ and $\KTnb$.

\subsection{Our results} 
			
			Our main results establish penalties for comparable knowledge types in the partial order $\ll$. For knowledge types with known distance, the penalty for having a blind map over a complete map turns out to be arbitrarily large: more precisely, it cannot be bounded by any function of the radius $m$. By contrast, for unknown distance, the penalty of having a blind map over having a complete map is small: it is $O(1)$. When a map is provided (either complete or blind), the penalty of not knowing the distance over knowing it is medium: we prove that it is linear in the radius $m$, and that this cannot be improved, i.e., it cannot be bounded by any sublinear function of this radius. The same  remains true for the penalty between the extreme types: $\KTnb$ over $\KTdc$. Figure \ref{fig_kt_order} depicts these results.

			It is interesting to note that non-constant penalties between various knowledge types are in fact due to the presence of leaves at different levels.
This is implied by the observation that if we restrict attention to the class of trees where all leaves are at the same (last) level, then all penalties
between comparable knowledge types become $O(1)$ (see the justification in Section~\ref{ccl}).

			\begin{figure}
				\centering
				\includegraphics[width=8cm]{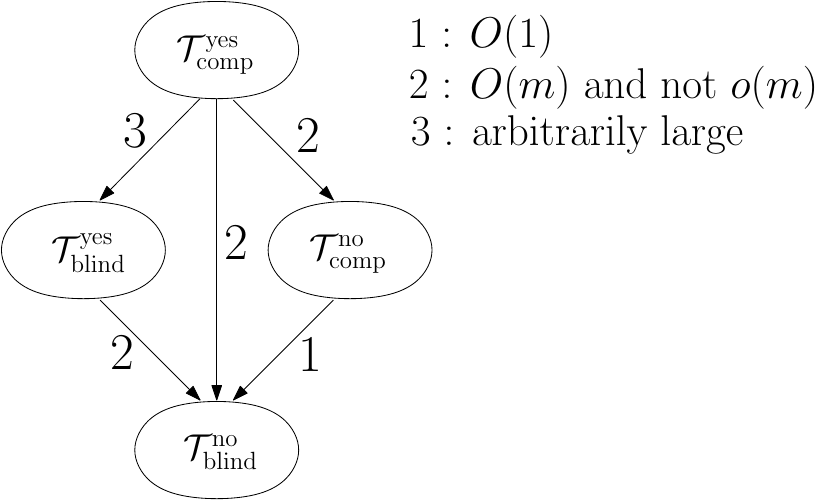}
				\caption{Penalties between different knowledge types\label{fig_kt_order}}
			\end{figure}

\subsection{Related work} 

			The problem of searching for a target by one or more mobile agents (often called treasure hunt) has a long history dating back to the sixties and seventies \cite{Bec1964,BN1970,Bel1963}. In these papers, the authors considered treasure hunt strategies for random locations of the treasure on the line. The book \cite{AG} surveys both the search for a fixed target and the related rendezvous problem, where the target and the searching agent are both mobile and their role is symmetric: they cooperate to meet. This book is concerned mostly with randomized search strategies. In \cite{MP,TSZ}, the authors studied relations between the problems of treasure hunt (searching for a fixed target) and rendezvous in graphs. In  \cite{CDGKSU}, the authors considered the task of searching for a non-adversarial, uncooperative mobile agent moving with constant speed on a cycle, by a team of cooperating mobile robots.
			The authors of \cite{BCR} studied the task of finding a fixed point on the line and in the grid, and initiated the study of the task of searching for an unknown line in the plane. This research was continued, e.g., in \cite{DKP,JL,La2}. In \cite{KR,SF}, the authors concentrated on game-theoretic aspects of the situation where multiple selfish pursuers compete to find a target.

			Several papers considered  treasure hunt in the plane, see surveys \cite{Gal2013,GK2010}. The main result of \cite{La} is an optimal algorithm to sweep a plane in order to locate an unknown fixed target, where locating means to get the agent originating at point $O$ to a point $P$ such that the target is in the segment $OP$. In \cite{FHGTM}, the authors considered the generalization of the search problem in the plane to the case of several searchers.  Efficient search for a fixed or a moving target in the plane, under complete ignorance of the searching agent, was studied in \cite{Pe}. In \cite{BDPP}, the authors studied treasure hunt in the plane with the help of interactive hints indicating an angle in which the treasure is located.

			Treasure hunt on the line (also called the cow-path problem \cite{KRT}) has been also generalized to the environment consisting of multiple rays originating at a single point \cite{AAD,DFG2006,LS2001,Sch2001}. The first two of these papers analyzed treasure hunt with turn costs, both in the line and in multiple rays. The authors of \cite{LS2001} consider the variant of treasure hunt in multiple rays of bounded length.

			Treasure hunt in arbitrary graphs has been studied, e.g., in \cite{ABRS,DKK2006}. The authors of \cite{ABRS,DKK2006} were mainly interested in efficient constrained exploration strategies. However, they got interesting corollaries for treasure hunt: they designed algorithms finding a treasure at distance $d$ at a cost linear in the number of edges of a ball with radius only slightly larger than $d$. The authors of \cite{ABRS} conjectured that it is impossible to find a treasure at cost nearly linear in the number of edges of a ball with radius exactly $d$. In the recent paper \cite{BDLP} this conjecture was refuted: it was shown that this is in fact possible with only logarithmic overhead on this linear cost.
			
			In \cite{FKK+2008}, the authors considered treasure hunt in several classes of graphs including trees. Treasure hunt in trees was the object of study in \cite{DCD,DCS1995,KZ2011}. In the first two of these papers, the authors restricted attention to complete $b$-ary trees, and in the third one, treasure hunt was studied in symmetric trees but admitting multiple treasures.

			In \cite{KKKS,MP}, treasure hunt in graphs was considered under the advice paradigm, where a given number of bits of advice, rather than a specified piece of knowledge, can be given to the agent, and the issue is to minimize this number of bits. For the related problem of graph exploration, the comparison of the impact of different types of maps on the efficiency of exploration was considered in \cite{DP}. The impact of different types of knowledge on the efficiency of the treasure hunt problem restricted to symmetric trees was studied in \cite{KZ2011}. (Note that in such trees, having a complete or a blind map is equivalent). To the best of our knowledge this impact has never been studied before for arbitrary trees. This is the subject of the present paper.

	\section{Complete map vs. blind map}

		\subsection{Known distance}

			We start with a proposition showing that if the distance from the root to the treasure is known then the penalty of having a blind map over a complete map cannot be bounded by any function of the radius $m$.

			\begin{proposition}\label{lem_good_bad_dist}
				There is no function $f:\mathbb{N}\rightarrow\mathbb{N}$ such that $\Penalty(\KTdb/\KTdc)$ is $O(f)$.
			\end{proposition}

			\begin{proof}
				Assume for contradiction that there exists a function $f:\mathbb{N}\rightarrow\mathbb{N}$ with the property that $\Penalty(\KTdb/\KTdc)$ is $O(f)$. Hence there is a constant $c>0$ and a deterministic algorithm $\cA_2$ such that for any algorithm $\cA_1$, we have:
				$$\forall K_2\in \KTdb,\forall m\geq 1, \overhead_{\cA_2}^{K_2}(m) \leq c \cdot f(m) \cdot \max_{\substack{K_1\in \KTdc\\ K_1\subseteq K_2}}\left(\overhead_{\cA_1}^{K_1}(m)\right).$$

			\begin{figure}
				\centering
				\includegraphics[width=5cm]{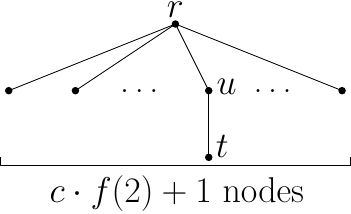}
				\caption{The tree $S$\label{fig_tree_S}}
			\end{figure}
				
				Let $S$ be a tree with a root $r$, with $c \cdot f(2)+1$ nodes at distance 1 from $r$, and with one node $t$ at distance 2 from $r$ attached to some node, called $u$ (see Figure~\ref{fig_tree_S}). Let $\mathcal{S}$ be the set of all root-preserving isomorphic copies of $S$. Let $K_2=\{(T,2)\mid T\in  \cS\}$. We have $B_2(K_2)=K_2\neq \emptyset$, since the distance is equal to 2 in every instance of $K_2$.  Observe that the node $t$ is the only possible location of the treasure in instances of $K_2$. Let $\cA_2$ be any deterministic algorithm that is given knowledge $K_2$. This algorithm visits the children of $r$ by taking port numbers at $r$ in some fixed order. There exists $W\in\cS$ such that $u$ is the last child of $r$ visited by $\cA_2$. The cost of $\cA_2$ on instance $(W, 2)$ is at least $2c \cdot f(2) + 2$ since for each child of $r$ except $u$, $\cA_2$ visits this child, comes back to $r$ and then moves at least twice to reach $t$ from $r$. Hence, we have:
				$$\overhead_{\cA_2}^{K_2}(2) = \max_{(T, d)\in B_2(K_2)}\frac{C_{\cA_2}^{K_2}(T, d)}{d}=\max_{T\in \cS}\frac{C_{\cA_2}^{K_2}(T, 2)}{2}\geq \frac{C_{\cA_2}^{K_2}(W, 2)}{2}\geq \frac{2c \cdot f(2) + 2}{2}= c \cdot f(2)+1.$$
				Any $K_1\in \KTdc$ such that $K_1\subseteq K_2$ is a singleton $\{(T, 2)\}$ with $T\in \cS$. Using any such knowledge $K_1$, the agent can find the port number to take at $r$ in order to get to $u$. There exists an algorithm $\cA_1$ that given any $K_1$ can move from $r$ to $u$ and then move to $t$ in two steps. Hence, we have:
				$$\max_{\substack{K_1\in \KnowType_1\\ K_1\subseteq K_2}}\left(\overhead_{\cA_1}^{K_1}(2)\right) = \max_{(T, d)\in B_2(K_1)}\frac{C_{\cA_1}^{K_1}(T, d)}{d} = \frac{2}{2}=1.$$
				Hence, for any algorithm $\cA_2$, there exists an algorithm $A_1$ such that:
				$$\exists K_2\in \KTdb, \overhead_{\cA_2}^{K_2}(2) \geq c \cdot f(2)+1 >c \cdot f(2) \cdot \max_{\substack{K_1\in \KnowType_1\\ K_1\subseteq K_2}}\left(\overhead_{\cA_1}^{K_1}(2)\right).$$
				This contradiction proves the proposition. Note that, according to Definition~\ref{def}, it is enough to show the contradiction for some $m\geq 1$, and we chose to give a counterexample for $m=2$. However, a similar conterexample could be given for any $m\geq 2$.
			\end{proof}

		\subsection{Unknown distance}

			In this section, we show that if the distance to the treasure is not known to the agent, then the penalty of having a blind map over having a complete map is only $O(1)$. In order to show this, we design a treasure hunt algorithm for an agent given only a blind map. The idea of the algorithm is to perform a sequence of DFS explorations at increasing levels, where the level of the next DFS is computed using the blind map. We need to preserve a compromise. On the one hand we want to increase the levels sufficiently fast in order to ``telescope'' the overall cost of iterations, but on the other hand we don't want to increase them too fast because this could cause visiting too many nodes  before reaching the treasure. It turns out that a good compromise is the following. Suppose that the agent visited all nodes at levels up to $h$. Let $k>h$ be the minimum level such that $L_{h+1}^k$ is at least $ L_1^h$. If  $L_{h+1}^k$ is at least 3 times larger than $ L_1^h$ and $k>h+1$ then we take $k-1$ as the next DFS level, otherwise we take $k$. We denote by $DFS_h$ the DFS performed up to level $h$, called the level of this DFS, by visiting children in increasing order of port numbers at each node, and returning to the root at the end.

			\vspace*{0.5cm}

			\begin{algorithm}[H]
				\caption{Treasure Hunt with a Blind Map\label{alg_bad_map}}
				\SetKwInOut{Input}{input}
				\Input{A blind map of the tree}
				$h:= 1$\;
				\While{the treasure is not found}{
					Perform $DFS_h$\;
					$k:= \min(i\in\mathbb{N} \mid i \geq h+1 \mbox{ and }  L_{h+1}^i\geq L_1^h)$\;
					\eIf{$L_{h+1}^{k}\geq 3L_{1}^{h}$ and $k\geq h+2$}{
						$h:= k-1$\;
					}{
						$h:= k$
					}
				}
			\end{algorithm}

			\vspace*{0.5cm}

			The following lemma shows that Algorithm~\ref{alg_bad_map} has a cost at most a constant times larger than the number of nodes at distance at most $d$ from the root, for any instance $(T,d)$. It will be a key piece to show that the penalty of having a blind map over having a complete map is $O(1)$, since this number of nodes over $d$ is a lower bound for any treasure hunt algorithm not knowing the distance.

			\begin{lemma}\label{lem_upper_bound_algo}
				For any $K\in \KTnb$, Algorithm~\ref{alg_bad_map} with knowledge $K$ has a cost at most $16L_1^d(T)$, for any instance $(T,d)\in K$.
			\end{lemma}

			\begin{proof}
				Let $h_i$ be the level of the $i$th DFS of Algorithm~\ref{alg_bad_map}, i.e., the value of variable $h$ at the start of the $i$th iteration of the loop. Let $k_i$ be the value of $k$ at the end of the $i$th iteration of this loop. Let $b_i$ be the truth value of the if condition at line 5 during the $i$th iteration of the loop, i.e., the truth value of the conjunction $L_{h_{i}+1}^{k_i}\geq 3L_{1}^{h_i}$ and $k_i\geq h_i+2$.

				First, we show that $h_i$ is a strictly increasing sequence. For each $i$, we consider two cases depending on the value of $b_i$. If $b_i$ is false, we have $h_{i+1}= k_i\geq h_{i}+1$. If $b_i$ is true, then we have $k_i\geq h_i+2$ and $h_{i+1}= k_i-1\geq h_i+1$. Hence, in all cases we have $h_{i+1}\geq h_i+1$.

				Let $l$ be the minimum value such that $h_l\geq d$. The treasure will be found during the $l$th iteration of the loop. We prove the following claim that intuitively shows that the size of levels explored by the algorithm grows fast enough. It will be useful to show that the total cost of all iterations of the loop is proportional to the cost of the exploration of the last level.

				\begin{claim} \label{claim}
					For any $1<i<l$, either $b_{i-1}$ is false and $L_1^{h_i}\geq 2L_1^{h_{i-1}}$, or $b_{i-1}$ is true and the following conditions hold: $L_1^{h_{i+1}}\geq 4L_1^{h_{i-1}}$, $L_{1}^{h_i}< 2L_{1}^{h_{i-1}}$, $L_1^{h_{i+1}}\geq 2L_1^{h_{i}}$, and $b_i$ is false.
				\end{claim}

				In order to prove the claim, first consider the case when $b_{i-1}$ is false. The value $h_i$ is then equal to $k_{i-1}$. We have $L_{h_{i-1}+1}^{k_{i-1}}\geq L_{1}^{h_{i-1}}$ by the choice of $k_{i-1}$. Hence, we have $L_{h_{i-1}+1}^{h_{i}}=L_{h_{i-1}+1}^{k_{i-1}}\geq L_{1}^{h_{i-1}}$ and so $L_{1}^{h_{i}} = L_{1}^{h_{i-1}}+L_{h_{i-1}+1}^{h_{i}} \geq 2L_{1}^{h_{i-1}}$.

Next, consider the case when $b_{i-1}$ is true. In this case, the value $h_i$ is equal to $k_{i-1}-1$. Since $k_{i-1}$ is the smallest value $q$ such that $q \geq h_{i-1}+1 \mbox{ and } L_{h_{i-1}+1}^{q}\geq L_{1}^{h_{i-1}}$, we have $L_{h_{i-1}+1}^{h_i}< L_{1}^{h_{i-1}}$ and so $L_{1}^{h_i}< 2L_{1}^{h_{i-1}}$. Since $b_{i-1}$ is true, we have $L_{h_{i-1}+1}^{k_{i-1}}\geq 3L_{1}^{h_{i-1}}$. We can deduce that $L_{1}^{k_{i-1}}\geq 4L_{1}^{h_{i-1}}\geq 2L_{1}^{h_{i}}$ and so $L_{h_i+1}^{k_{i-1}}\geq L_{1}^{h_{i}}$. It means that $k_{i} = k_{i-1}$. We have $k_{i}= k_{i-1}=h_{i}+1$, hence $b_{i}$ is false which implies $h_{i+1}=k_{i}=k_{i-1}$. Hence, we have $L_{1}^{h_{i+1}}\geq 4L_{1}^{h_{i-1}}$ and $L_{1}^{h_{i+1}}\geq 2L_{1}^{h_{i}}$. This proves the claim.

				Observe that $2L_1^{h_i}$ is exactly the cost of the $DFS_{h_i}$, since $L_1^{h_i}$ is equal to the number of edges whose both endpoints are at distance at most $h_i$ from the root. Let $C_i = \sum_{j=1}^i 2L_1^{h_j}$ be the total cost of the algorithm up to the $i$th iteration of the loop. 
				We prove the following claim.
				
				\begin{claim} \label{claim2}
				For any $i\geq 1$, the following inequality named $P(i)$ holds:
				$$
				C_i \leq  \left\{
				\begin{array}{l}
				4L_1^{h_i}\mbox{ if $i=1$ or $b_{i-1}$ is false,}\\
				6L_1^{h_i}\mbox{ otherwise.}
				\end{array}
				\right.
				$$
				\end{claim}

			We will show by induction on $i\geq 1$ that the inequality $P(i)$ holds for all $i\geq 1$.
				Observe that $P(1)$ is true since $C_1 = 2L_1^{h_1}$ and $P(2)$ is true since $C_2=2L_1^{h_1} + 2 L_1^{h_2}\leq 4L_1^{h_2}$. For $i\geq 2$, assume that $P(j)$ is true for all $j\leq i$. We will show that $P(i+1)$ is true. We need to consider two cases depending on the truth value of $b_i$.

				\begin{itemize}
					\item \textbf{Case 1:} $b_i$ is true\\
						In this case, $b_{i-1}$ must be false in view of Claim~\ref{claim}. Hence, we have $C_i\leq 4L_1^{h_i}$ since $P(i)$ holds. This implies:
						$$C_{i+1} = C_i + 2L_1^{h_{i+1}}\leq 4L_1^{h_i} + 2L_1^{h_{i+1}} \leq 6L_1^{h_{i+1}}.$$
					\item \textbf{Case 2:} $b_i$ is false\\
						We consider two subcases depending on the truth value of $b_{i-1}$:
						\begin{itemize}
							\item \textbf{Subcase 2.1:} $b_{i-1}$ is false\\
								In this case, we have $C_i\leq 4L_1^{h_i}$ since $P(i)$ holds, and, in view of Claim~\ref{claim}, $L_1^{h_{i+1}}\geq 2L_1^{h_{i}}$ since $b_i$ is false. This implies:
								$$C_{i+1} = C_i + 2L_1^{h_{i+1}}\leq 4L_1^{h_i} + 2L_1^{h_{i+1}} \leq 4L_1^{h_{i+1}}.$$
							\item \textbf{Subcase 2.2:} $b_{i-1}$ is true\\
								In this case, we have $L_1^{h_{i+1}}\geq 4L_1^{h_{i-1}}$ and $L_1^{h_{i+1}}\geq 2L_1^{h_{i}}$, in view of Claim~\ref{claim}. Observe that since $b_{i-1}$ is true, either $i=2$ or $b_{i-2}$ must be false in view of Claim~\ref{claim}. This means that $C_{i-1}\leq 4L_1^{h_{i-1}}$ since $P(i-1)$ holds. This implies:
								\begin{eqnarray*}
									C_{i+1}	& = &	C_{i-1} + 2L_1^{h_{i}} + 2L_1^{h_{i+1}}\\
											&\leq&	4L_1^{h_{i-1}} + L_1^{h_{i+1}} + 2L_1^{h_{i+1}} \mbox{ (since }C_{i-1}\leq 4L_1^{h_{i-1}} \mbox{ and } L_1^{h_{i+1}}\geq 2L_1^{h_{i}})\\
											&\leq&	L_1^{h_{i+1}} + 3L_1^{h_{i+1}} \mbox{ (since }L_1^{h_{i+1}}\geq 4L_1^{h_{i-1}})\\
											&=&		4L_1^{h_{i+1}}.
								\end{eqnarray*}
						\end{itemize}
				\end{itemize}

				This proves the inequality $P(i)$ by induction and ends the proof of the claim.

				It remains to show that $C_l \leq 16L_1^{d}$. Recall that $l$ is the smallest integer such that $h_l\geq d$. If $h_l=d$ then $C_l\leq 6L_1^{h_{l}}\leq 16L_1^d$. If $h_l>d$ then we have $h_{l-1} + 1 < h_{l}$ since $h_{l-1}<d$. We consider two cases depending on the truth value of $b_{l-1}$. If $b_{l-1}$ is false then $k_{l-1}=h_l\geq h_{l-1}+2$. Since the conjunction ($L_{h_{l-1}+1}^{k_{l-1}}\geq 3L_{1}^{h_{l-1}}$ and $k_{l-1}\geq h_{l-1}+2$) is false, we have $L_{h_{l-1}+1}^{h_l} < 3L_{1}^{h_{l-1}}$ and so $L_{1}^{h_l} < 4L_{1}^{h_{l-1}}$. This implies:
				$$C_{l}\leq 4L_{1}^{h_l} < 16L_{1}^{h_{l-1}} < 16L_1^d.$$
				If $b_{l-1}$ is true then we have $L_{1}^{h_l} < 2L_{1}^{h_{l-1}}$, in view of Claim~\ref{claim}. This implies
				$$C_{l}\leq 6L_{1}^{h_l} < 12L_{1}^{h_{l-1}} < 12L_1^d.$$
				Hence, in all cases, we have $C_{l} \leq 16L_1^d$.
			\end{proof}

			The next lemma gives a lower bound on the overhead of any algorithm that is given only a complete map.

			\begin{lemma}\label{lem_lower_bound}
				For any deterministic algorithm $\cA$, for any $K\in \KTnc$ and for any integer $m\geq 1$, we have:
				$$\overhead_{\cA}^{K}(m) \geq \max_{(T,d)\in B_m(K)}\frac{L_1^d(T)}{d}.$$
			\end{lemma}

			\begin{proof}
				Let $K$ be any knowledge in knowledge type $\KTnc$ and let $\cA$ be any deterministic treasure hunt algorithm. There exists a tree $S$ such that $K=\{(S,d)\mid 1\leq d\leq h(S)\}$. Let $m$ be any positive integer. We have $B_m(K)=\{(S,d)\mid 1\leq d\leq \min(h(S),m)\}$. Let $l$ be any integer such that $1\leq l \leq \min(h(S),m)$ and let $U_1^l$ be the union of all levels of $S$ from one to $l$. Let $t$ be the last node of $U_1^l$ visited by $\cA$. Let $y$ be the distance between the root and $t$. In the instance $(S,y)$, the treasure can be hidden in $t$ and so $C_{\cA}^K(S, y)$ is equal to the number of moves performed by $\cA$ before reaching $t$. Since $\cA$ has visited all nodes of $U_1^l$ when it reaches $t$, $C_{\cA}^K(S, l)$ is at least the number of nodes of $U_1^l$ which is $L_1^l(S)$. Hence, we have:
				$$\overhead_{\cA}^K(m) = \max_{(T,d)\in B_m(K)}\frac{C_{\cA}^K(T,d)}{d}= \max_{1\leq d\leq \min(h(S),m)}\frac{C_{\cA}^K(S,d)}{d} \geq \frac{C_{\cA}^K(S,y)}{y} \geq \frac{L_1^l(S)}{y}\geq \frac{L_1^l(S)}{l}.$$
				Since this inequality is true for any $1\leq l \leq \min(h(S),m)$, this implies:
				$$\overhead_{\cA}^{K}(m) \geq \max_{1\leq d\leq \min(h(S),m)}\frac{L_1^d(S)}{d}=\max_{(T,d)\in B_m(K)}\frac{L_1^d(T)}{d}.$$
			\end{proof}

			The final theorem of this section shows that the penalty of having only a blind map over having only a complete map is small.

			\begin{theorem}
				$\Penalty(\KTnb/\KTnc)$ is $O(1)$.
			\end{theorem}

			\begin{proof}
				Let $K_2$ be any knowledge in knowledge type $\KTnb$. There exists a tree $S$ such that all trees in instances of $K_2$ are root-preserving isomorphic copies of $S$. Let $\cA_2$ be Algorithm~\ref{alg_bad_map} and let $m$ be any positive integer. We have:
				$$
				\overhead_{\cA_2}^{K_2}(m)=\max_{(T,d)\in B_m(K_2)}\frac{C_{\cA_2}^{K_2}(T,d)}{d}\leq\max_{1\leq d\leq \min(h(S),m)}\frac{16L_1^d(S)}{d} \mbox{ by Lemma \ref{lem_upper_bound_algo}.}
				$$

				In view of Lemma \ref{lem_lower_bound}, for any algorithm $\cA_1$ and for any $K_1\in \KTnc$ such that $K_1\subseteq K_2$, we have:
				$$\overhead_{\cA_1}^{K_1}(m) \geq \max_{(T, d)\in B_m(K_1)} \frac{L_1^d(T)}{d}=\max_{1\leq d\leq\min(h(S),m)}\frac{L_1^d(S)}{d}.$$
				This implies:
				$$
				16\cdot\max_{\substack{K_1\in \KTnc\\ K_1\subseteq K_2}}\left(\overhead_{\cA_1}^{K_1}(m)\right)\geq \max_{1\leq d\leq\min(h(S),m)} \frac{16L_1^d(S)}{d}\geq \overhead_{\cA_2}^{K_2}(m).
				$$
			\end{proof}

	\section{The impact of knowing the distance}

		In this section we show that the penalty of not knowing the distance to the treasure over knowing it, both when a complete map and when a blind map is given to the agent, is medium with respect to penalties considered in the previous section. More precisely, we prove that this penalty is linear in the radius $m$, and that this estimate is tight in the sense that this penalty is not any sublinear function of $m$.

		The first lemma shows that the above mentioned penalties, as well as the penalty of the least precise knowledge type $\KTnb$ over the most precise knowledge type $\KTdc$ are linear in $m$.

		\begin{lemma}\label{up}
			$\Penalty(\KTnc/\KTdc)$, $\Penalty(\KTnb/\KTdb)$ and $\Penalty(\KTnb/\KTdc)$ are $O(m)$.
		\end{lemma}

		\begin{proof}
			For any algorithm $\cA_1$ that is given knowledge $K_1\in\KTdc\cup\KTdb$, we consider the cost $C_{\cA_1}^{K_1}(T,d)$ for finding the treasure at distance $d$ in tree $T$,  where $(T, d)$ is any instance in $K_1$. In order to visit all nodes at distance exactly $d$, the agent must first move to one of those nodes, at a cost at least $d$. Then, it has to move to the $l_d(T)-1$ other nodes at distance $d$, at a cost at least $2(l_d(T)-1)$ since such nodes are at distance at least 2 from each other. Thus the total cost of any algorithm $\cA_1$ for finding the treasure at distance $d$ in tree $T$ is at least $2(l_d(T)-1)+d\geq l_d(T)$, for every $(T, d)\in K_1$.

			We now consider Algorithm~\ref{alg_bad_map} and denote it by $\cA_2$. Notice that, although the input of this algorithm is a blind map, we can always give a complete map as its input, and let the algorithm ignore the port numbers. In view of Lemma~\ref{lem_upper_bound_algo}, whether the knowledge $K_2$ given to $\cA_2$ belongs to $\KTnc$ or to $\KTnb$, the cost $C_{\cA_2}^{K_2}(T,d)$ is at most $16L_1^d(T)$, for every $(T, d)\in K_2$.

			Let $K_2$ be any knowledge in $\KTnc\cup\KTnb$, and let $S$ be any tree such that all trees in instances of $K_2$ are root-preserving isomorphic copies of $S$. For any integer $m\geq 1$, we have:
			\begin{eqnarray*}
				\overhead_{\cA_2}^{K_2}(m)	& \leq &	\max_{(T,d)\in B_m(K_2)}\left(\frac{16L_1^d(T)}{d}\right)\\
											& \leq &	16\max_{1\leq d\leq \min(m, h(S))}\left(\frac{l_1(S)}{d} + \frac{l_2(S)}{d} + \frac{l_3(S)}{d} + \dots + \frac{l_{d}(S)}{d}\right)\\
											& \leq &	16\max_{1\leq d\leq \min(m, h(S))}\left(\frac{l_1(S)}{1} + \frac{l_2(S)}{2} + \frac{l_3(S)}{3} + \dots + \frac{l_{d}(S)}{d}\right)\\
											& \leq &	16\max_{1\leq d\leq \min(m, h(S))}\left(\frac{C_{\cA_1}^{K_1}(S,1)}{1} + \frac{C_{\cA_1}^{K_1}(S,2)}{2} + \frac{C_{\cA_1}^{K_1}(S,3)}{3} + \dots + \frac{C_{\cA_1}^{K_1}(S,d)}{d}\right)\\
											& \leq &	16\max_{1\leq d\leq \min(m, h(S))}\left(d\max_{1\leq i\leq d}\left(\frac{C_{\cA_1}^{K_1}(S,i)}{i}\right)\right)\\
											& \leq &	16m\max_{1\leq d\leq \min(m, h(S))}\left(\frac{C_{\cA_1}^{K_1}(S,d)}{d}\right).
			\end{eqnarray*}
			Hence, for any integer $m\geq 1$ and for any $\KnowType_1\in\{\KTdc,\KTdb\}$, we have:
			$$\overhead_{\cA_2}^{K_2}(m)\leq 16m\cdot\max_{\substack{K_1\subseteq K_2\\K_1\in\KnowType_1}}\left(\overhead_{\cA_1}^{K_1}(m)\right).$$
			This implies that all three penalties mentioned in the Lemma are $O(m)$.
		\end{proof}

		The following lemma shows a lower bound on the penalties of not knowing the distance to the treasure over knowing it, for algorithms that are given a blind map or a complete map.

		\begin{lemma}\label{down}
			There is no function $f(m)$ in $o(m)$ with the property that $\Penalty(\KTnc/\KTdc)$, $\Penalty(\KTnb/\KTdb)$ or $\Penalty(\KTnb/\KTdc)$ is $O(f(m))$.
		\end{lemma}

		\begin{proof}
			For any $\ell\geq 2$, let $P_\ell$ be the path of length $\ell$ with vertices $u_0, u_1, u_2, \dots, u_\ell$. Consider the following tree $T_\ell$ with any port numbering. $T_\ell$ is the tree rooted at $u_0$ obtained from $P_\ell$ by adding a node $v_{i+1}$ attached to $u_i$ and attaching $i+3$ leaves to this node $v_{i+1}$, for each $0\leq i\leq \ell-2$ (see Figure \ref{fig_caterpillar_bis}). Clearly, we have $h(T_\ell)=\ell$.

			\begin{figure}
				\centering
				\includegraphics[width=10cm]{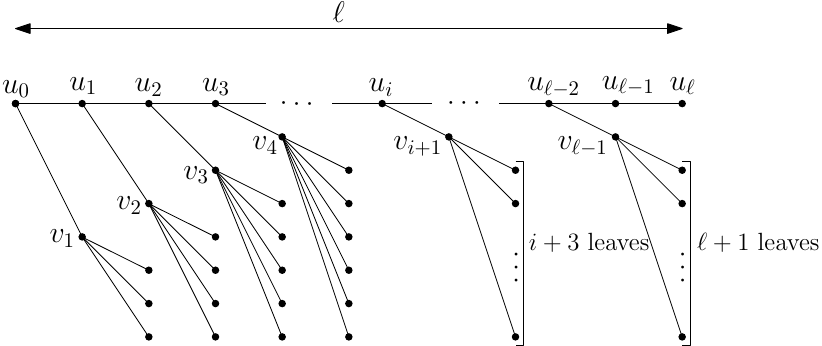}
				\caption{A tree $T_{\ell}$\label{fig_caterpillar_bis}}
			\end{figure}

			For any $\ell\geq 2$, let $\cS_\ell$ be the family of all sets $K\in \KTnb\cup\KTnc$ such that for all instances $(T, d)\in K$, $T$ is a root-preserving isomorphic copy of $T_\ell$. Let $K_2^\ell$ be any knowledge in $\cS_\ell$. Notice that $K_2^\ell$ contains instances $(T,d')$, for all $1 \leq d' \leq \ell$ and so $B_\ell(K_2^\ell)=K_2^\ell\neq \emptyset$. For any algorithm $\cA_2$ that is given knowledge $K_2^\ell$ and any $(S, d)\in K_2^\ell$, the behavior of algorithm $\cA_2$ on instance $(S, d)$ does not depend on $d$ since $\cA_2$ is given as input instances for all possible distances between the root and the position of the treasure. Let $u_S$ be the last node of $S$ visited by algorithm $\cA_2$. Let $x$ be the distance between $u_S$ and the root of $S$. The cost $C_{\cA_2}^{K_2^\ell}(S, x)$ is at least the number of nodes in $T_{\ell}$ minus 1 which is $1+\sum_{i=0}^{\ell-2}(i+5)=\frac{\ell^2+7\ell-6}{2}$. Hence, we have:
			$$\overhead_{\cA_2}^{K_2^\ell}(\ell) = \max_{(T,d)\in K_2^\ell}\frac{C_{\cA_2}^{K_2^\ell}(T, d)}{d}\geq \frac{C_{\cA_2}^{K_2^\ell}(S, x)}{x} \geq \frac{\ell^2+7\ell-6}{2x}\geq \frac{\ell^2+7\ell-6}{2\ell} \geq \frac{\ell+4}{2} \mbox{ since } \ell\geq 2.$$

			Let $K_1^\ell$ be any knowledge in $\KTdb\cup\KTdc$ such that $K_1^\ell\subseteq K_2^\ell$. We will describe an algorithm $\cA_1$ with knowledge $K_1^\ell$ such that $C_{\cA_1}^{K_1^\ell}(T, d)\leq 7d$ for any $(T,d) \in K_1^\ell$ and so $\overhead_{\cA_1}^{K_1^\ell}(\ell)\leq 7$. Without loss of generality, we can consider that $K_1^\ell\in \KTdb$. Indeed, any algorithm taking a blind map as input can be given a complete map as input and instructed to ignore the port numbers. Algorithm $\cA_1$ is given the set of instances $(T,d)$, where $d$ is fixed, and trees $T$ are all root-preserving isomorphic copies of $T_\ell$ with any port numbering. We consider two cases depending on the value of $d$. 

			\begin{itemize}
					\item \textbf{Case 1:} $d=1$\\
			In this case, the agent visits the two children of $u_0$ in three steps and we have $C_{\cA_1}^{K_1^\ell}(T, d)\leq 7d$, for any $(T,d) \in K_1^\ell$.
		
					\item \textbf{Case 2:} $d\geq 2$\\
			In this case, the agent moves to $u_{d-2}$ and then visits all descendants of $u_{d-2}$ at distance 2. We first show that for any $0\leq i\leq \ell-2$, the agent can move from $u_{i}$ to $u_{i+1}$ in at most 3 moves. If the agent is at $u_i$, it knows the port leading to the parent of $u_{i}$ and can take one of the two ports leading to a child of $u_{i}$. If this child has degree 3, then it is $u_{i+1}$ since the other child $v_{i+1}$ has a degree at least 4. If the degree of the child is more than 3, then this child is $v_{i+1}$ and the agent comes back to $u_i$ and then moves to the last remaining child $u_{i+1}$ of $u_i$. In all cases, the agent can reach $u_{i+1}$ from $u_{i}$ in at most 3 moves. Hence, it can move from $u_0$ to $u_{d-2}$ in at most $3d-6$ moves. The agent can then reach all nodes at distance $d$ from $u_0$ by performing $DFS_2$ (recall that this is DFS at depth 2)  on the subtree $S$ rooted at $u_{d-2}$. This $DFS_2$ costs at most $2(d+5)=2d+10$ since $S$ has at most $d+5$ edges. For any $(T,d)\in K_1^\ell$, the total cost $C_{\cA_1}^{K_1^\ell}(T, d)$ is at most $3d-6+2d+10=5d+4\leq 7d$, since $d\geq 2$.

			\end{itemize}
 Hence, for any $K_1^\ell$, we have:
			\begin{equation}
			\overhead_{\cA_1}^{K_1^\ell}(\ell) = \max_{(T,d)\in B_\ell(K_1^\ell)} \frac{C_{\cA_1}^{K_1^\ell}(T, d)}{d}\leq \frac{7d}{d} = 7.
			\label{eqn:overhead}
			\end{equation}
			Assume for contradiction that there exists a function $f(m)$ in $o(m)$ such that one of the penalties $\Penalty(\KTnc/\KTdc)$, $\Penalty(\KTnb/\KTdb)$ or $\Penalty(\KTnb/\KTdc)$ is $O(f(m))$. Denote $\cT^{\dist}=\{\KTdb, \KTdc\}$ and $\cT^{\nodist}= \{\KTnb, \KTnc\}$.

			Hence, there exists a constant $c$ and an algorithm $\cA_2$ such that for any algorithm $\cA_1$, we have:
			$$\exists \cT_2\in \cT^{\nodist}~\exists \cT_1\in \cT^{\dist}, \cT_1\ll \cT_2, \forall K_2\in \cT_2~\forall m\geq 1, \overhead_{\cA_2}^{K_2}(m) \leq c\cdot f(m)\cdot\max_{\substack{K_1\in \cT_1\\ K_1\subseteq K_2}}\left(\overhead_{\cA_1}^{K_1}(m)\right).$$
			For any $\ell\geq 2$ and for any algorithm $\cA_2$ there exists an algorithm $\cA_1$ such that:

$$\forall K_2^\ell\in \cS_\ell, \overhead_{\cA_2}^{K_2^\ell}(\ell) \geq \frac{\ell + 4}{2} \geq \frac{\ell + 4}{14}\cdot\max_{\substack{K_1^\ell\in \KTdb\cup\KTdc\\ K_1^\ell\subseteq K_2^\ell}}\left(\overhead_{\cA_1}^{K_1^\ell}(\ell)\right)\mbox{ since }\overhead_{\cA_1}^{K_1^\ell}(\ell)\leq 7\mbox{ from (\ref{eqn:overhead}).}$$

			Since $f(m)$ is $o(m)$, there exists an integer $y$ such that $\frac{y + 4}{14}>c\cdot f(y)$. For any  $\cT_2\in \cT^{no}$, the set $\cT_2\cap\cS_y$ contains at least one element $K_2^y$. For any $\cT_1\in \cT^{yes}$ such that $\cT_1\ll \cT_2$, there exists $K_1^y\in \cT_1$ such that $K_1^y\subseteq K_2^y$. This implies that for all $\cT_2\in \cT^{no}$ and  for all $\cT_1\in \cT^{yes}$, such that $\cT_1\ll \cT_2$ we have:
			$$\exists K_2^y\in \cT_2, \overhead_{\cA_2}^{K_2^y}(y) > c\cdot f(y)\cdot\max_{\substack{K_1^y\in \cT_1\\ K_1^y\subseteq K_2}}\left(\overhead_{\cA_1}^{K_1^y}(y)\right).$$
			This contradiction proves the lemma.
		\end{proof}

		Lemmas \ref{up} and \ref{down} imply the following theorem.

		\begin{theorem}
			\begin{sloppypar}
				All penalties $\Penalty(\KTnc/\KTdc)$, $\Penalty(\KTnb/\KTdb)$ and $\Penalty(\KTnb/\KTdc)$ are $O(m)$ and this order of magnitude cannot be improved for any of them.
			\end{sloppypar}
		\end{theorem}

	\section{Conclusion} \label{ccl}

		We established penalties between all comparable pairs of knowledge states resulting from giving the agent either a complete or a blind map and giving it or not the distance to the treasure. As we have seen, these penalties differ significantly depending on the considered pair of knowledge types. As mentioned in the introduction, it turns out that these differences are caused by the presence of leaves at different levels of the tree. To see this, restrict attention to trees in which all the leaves are at the same (last) level of the tree. Call such trees {\em even}. It is easy to observe that for even trees, all considered penalties are $O(1)$. Indeed, these trees $T$ have the property that, for any treasure hunt algorithm $\cA$, regardless of the initial knowledge $K$ given to it, the cost $C_{\cA}^K(T,d)$ is at least $L_1^d(T)$ because every node at level less than $d$ has a descendant at level $d$. On the other hand, Lemma \ref{lem_upper_bound_algo} implies that there exists a treasure hunt algorithm $\cA'$ with the property that, given initial knowledge $K'$ of any of the four types, the cost $C_{\cA'}^{K'}(T,d)$ is at most $16L_1^d(T)$, for any instance $(T,d) \in K'$. This implies that all five penalties are constant.
		
		There are natural  additional knowledge types resulting from relaxing the information concerning the exact distance to the treasure. The agent could be given an interval $[d_1,d_2]$ to which the distance $d$ from the root to the treasure belongs. As previously, giving such an interval could be combined with giving the agent either a complete or a blind map. Clearly, if an interval $[d_1,d_2]$ is included in an interval $[d_1',d_2']$ then knowledge corresponding to giving the latter is less precise than that corresponding to giving the former.   It would be interesting to investigate the penalties between comparable states of knowledge resulting from such a relaxation.

		Another natural open problem left by our research is to extend the questions raised in this paper for trees to the class of arbitrary graphs.

\end{document}